\documentclass[11pt]{article}
\pdfoutput=1

\usepackage[utf8]{inputenc} % allow utf-8 input 
\usepackage[T1]{fontenc}    % use 8-bit T1 fonts
\usepackage[hidelinks]{hyperref}       % hyperlinks
\usepackage{url}            % simple URL typesetting
\usepackage{booktabs}       % professional-quality tables
\usepackage{amsfonts}       % blackboard math symbols
\usepackage{nicefrac}       % compact symbols for 1/2, etc.
\usepackage{microtype}      % microtypography
\usepackage{amsmath, amsthm, amssymb}
\usepackage{fancybox}
\usepackage{tkz-graph}
\usetikzlibrary{arrows}
\usepackage{subfig}
\usepackage[noend]{algpseudocode}
\usepackage{algorithm,algorithmicx}
\usepackage[left=1in, right=1in, top = 1in, bottom = 1in]{geometry}
\usepackage[numbers]{natbib}
\usepackage[symbol]{footmisc}
\usepackage{tikz}
\usepackage{enumitem}

\tikzset{
    every node/.style={
        circle,
        draw,
        fill          = black!50,
        inner sep     = 0pt,
        minimum width =4 pt
    }   
}

\newtheorem{thm}{Theorem}
\newtheorem{lem}{Lemma}
\newtheorem{cor}{Corollary}
\newtheorem{prob}{Problem}

\newtheorem{prop}{Proposition}
\newtheorem{rem}{Remark}

\newenvironment{reflemma}[1]{\medskip\parindent 0pt{\bf Lemma \ref{#1}.}\em }{\vspace{1em}}
\newenvironment{reftheorem}[1]{\medskip\parindent 0pt{\bf Theorem \ref{#1}.}\em }{\vspace{1em}}

\newenvironment{refprop}[1]{\medskip\parindent 0pt{\bf Proposition \ref{#1}.}\em }{\vspace{1em}}

\newcommand{\Sym}{\text{Sym}_n(\mathbb{R}_{\ge 0})}

\DeclareMathOperator*{\argmin}{arg\,min}

\DeclareCaptionSubType*{algorithm}

\title{Generalized Metric Repair on Graphs}

\date{}

\author{
  Anna C.~Gilbert \footnote{
  Department of Mathematics,
  University of Michigan - Ann Arbor,
  \texttt{annacg@umich.edu}}
  \and
  Rishi Sonthalia \footnote{
  Department of Mathematics,
  University of Michigan - Ann Arbor,
  \texttt{rsonthal@umich.edu}}
}

\begin{document}

\maketitle
\vspace{-2em}
\begin{abstract}
Many modern data analysis algorithms either assume that or are considerably more efficient if the distances between the data points satisfy a metric. These algorithms include metric learning, clustering, and dimensionality reduction. Because real data sets are noisy, the similarity measures often fail to satisfy a metric. For this reason, Gilbert and Jain~\cite{Gilbert2017} and Fan, et al.~\cite{Raichel2018} introduce the closely related problems of \emph{sparse metric repair} and \emph{metric violation distance}. The goal of each problem is to repair as few distances as possible to ensure that the distances between the data points satisfy a metric. 

We generalize these problems so as to no longer require all the distances between the data points. That is, we consider a weighted graph $G$ with corrupted weights $w$ and our goal is to find the smallest number of modifications to the weights so that the resulting weighted graph distances satisfy a metric. This problem is a natural generalization of the sparse metric repair problem and is more flexible as it takes into account different relationships amongst the input data points. As in previous work, we distinguish amongst the types of repairs permitted (decrease, increase, and general repairs). We focus on the increase and general versions and establish hardness results and show the inherent combinatorial structure of the problem. We then show that if we restrict to the case when $G$ is a chordal graph, then the problem is fixed parameter tractable. We also present several classes of approximation algorithms. These include and improve upon previous metric repair algorithms for the special case when $G=K_n$
\end{abstract}

%!TEX root = ./main.tex
\section{Introduction}

Similarity distances that satisfy a metric are fundamental to large number of machine learning tasks such as dimensionality reduction and clustering (see~\cite{Wang2015, Baraty2011} for two examples). However, due to noise, missing data, and other corruptions, in practice, these distances do not often adhere to a metric. Motivated by these observations and early work by \citet{Brickell}, \citet{Raichel2018} and, independently, \citet{Gilbert2017} respectively formulated the \emph{Metric Violation Distance (MVD)} and the \emph{sparse metric repair (SMR)} problems. Formally, the problem both authors studied was given a distance matrix, modify as few entries as possible so that the repaired distances satisfy a metric.

These algorithms were successfully used by Gilbert and Sonthalia~\cite{GilbertSonthalia:SwissRoll2018} to modify a large class of dimensionality reduction algorithms to can handle missing data. For many applications, (e.g. metric learning, metric embedding for graph metrics), however, we want to enforce restrictions on some of the distances and stipulate that the rest of the distances follow from these restrictions. For example, \citet{Tenenbaum2005} showed that {\sc Isomap} can be approximated by looking at only a few key points. In many of these cases, using the approximation algorithms in \cite{Raichel2018, Gilbert2017} may not produce meaningful results.

\begin{figure}[h]
\centering
\includegraphics[scale=0.35]{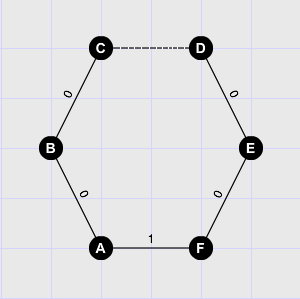} \hspace{3cm}
\includegraphics[scale=0.35]{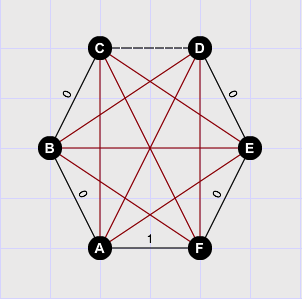} 

\caption{Example showing the blow up of optimal solution size, if we add in missing edges. Left: A $C_n$ with 1 edge of weight 1 and rest weight 0. The optimal solution has size 1. Right: The red edges have been added in and now the optimal solution has size $\Theta(n)$}
\label{fig:eg}

\end{figure}

As a more algorithmic example, let us consider traveling salesman problem (TSP). As \cite{Raichel2018} presented, this problem, in general, admits no polynomial time constant factor approximation \cite{TSP} (unless $P=NP$), but if we know the distances satisfy a metric, there is 1.5-approximation algorithm \cite{mTSP}. Thus, one approach to get an approximate solution to the general TSP $G$ could be to ``repair" the distances to satisfy a metric and then use the 1.5-approximation on the new weighted graph. In this case, if we naively complete the graph by setting the weight of the missing edges to be the shortest distance between its end points in $G$, then the size of the optimal solution may drastically increase, as shown in Figure \ref{fig:eg}. Thus, an approximation algorithm for repairing the distances may repair \emph{a large fraction of the original distances}. Hence, we would like an algorithm that repaired the metric by only changing the original given distances in $G$, and hence, giving us a tighter approximation ratio. 

To capture this more general nature of the problem we define the \emph{graph metric repair} problem as the natural graph theoretic generalization of the MVD and SMR problems: \begin{quote} Given a weighted graph $G$ and a set $\Omega$, find the smallest set of edges $S$, such that we can modify the weights of the edges in $S$ by values in $\Omega$ so that the new distances adhere to a metric.\end{quote}

This additional graph structure introduced in the generalized problem lets us incorporates different types of relationships amongst data points gives us more flexibility in its structure and hence, avails itself to be applicable to a richer class of problems. Furthermore, while \citet{Gilbert2017} showed that sparse metric repair can be approximated empirically via convex optimization, both \cite{Gilbert2017} and \cite{Raichel2018} developed combinatorial, as opposed to convex optimization, algorithms, based upon All Pairs Shortest Path (APSP) computations. Thus, metric repair is inherently a combinatorial problem and the generalized graph problem helps elucidate this structure. This insight and generalization helps us understand the combinatorial structure of the problem, carry out fine-grained complexity analysis that incorporates different types of relationships amongst data points, and analyze approximation algorithms.  

\noindent\textbf{Contributions and Results:} The main contributions of this paper are as follows: \begin{itemize}
\item We show that, the decrease only version of the problem ($\Omega = \mathbb{R}_{\le 0}$) can be solved in cubic time and that if we are allowed to increase distances at all, even by a single number, then the problem becomes NP-Complete. (Section 2.3.1)
\item We provide a \emph{characterization} for the the support of solutions to the increase ($\Omega = \mathbb{R}_{\ge 0}$) and general ($\Omega = \mathbb{R}$) versions of the problem. Furthermore, we provide a cubic time algorithm that determines for any subset of the edges whether there exists \emph{any valid solution} with that support and finds one (if it exits). (Section 2.3.2)
\item We present new insights that allow us to show that if we restrict to chordal graphs $G$ then problem is \emph{fixed parameter tractable}, thus answering an open question posed by~\cite{Raichel2018}. (Section 3)
\item We give a \emph{new approximation algorithm} for the general problem that has an approximation ratio of $L$ where $L+1$ is the length of the largest broken cycle and provide a comparative analysis of all three algorithms for the MVD or SMR problem. In particular, we improve the running time of the \citet{Raichel2018} algorithm from $\Theta(n^6)$ to $\Theta(n^5)$, and provide a lower bound for the approximation ratio for the algorithm from \cite{Gilbert2017}. (Section 4)
\end{itemize}

%!TEX root = ./main.tex
\section{Preliminaries}
\label{sec:prelim}

\subsection{Problem set up}

Let us start by defining some terminology. A cycle $C = v_1, \hdots, v_n$ in a weighted graph $G = (V,w)$ is \textbf{broken} if there exists an edge $v_iv_{i+1}$ such that \[ w(v_{i},v_{i+1}) > \sum_{e \in C \backslash \{v_i,v_{i+1}\}} w(e) \] In this case we shall call the edge on the left with the (too) large edge weight the \textbf{top} edge and the rest of the edges as \textbf{bottom} edges. Hence, given a weighted graph $G = (V,E,w)$ we say that $G$ satisfies a \textbf{metric} if there are no broken cycles. Finally, let {\rm Sym}$_n(\Omega)$ to be the set of $n \times n$ symmetric matrices with entries drawn from $\Omega$. Now we can define the generalized graph metric repair problem as follows: 

\begin{prob} \label{prob:metric} 
Given a set $\Omega$ and a weighted graph $G = (V,E,w)$ we want to find
\[ 
	\underset{W \in {\rm Sym}(\Omega)}{\argmin}\|W\|_0 \text{ such that } G = (V,E,w+W) \text{ satisfies a metric,} 
\] 
or return NONE, if no such $W$ exists. Here $\|W\|_0$ is number of non zero entries in $W$ or the $\ell_0$ pseudonorm entry-wise in the matrix $W$. Let us denote this problem as graph metric repair or MR($G, \Omega$). 
\end{prob}

We will also need the following basic graph theory definitions: $K_n$ is the complete graph on $n$ vertices. $C_n$ is the cycle $n$ vertices. A \textbf{chord} of a cycle is an edge connecting two non-adjacent vertices. A \textbf{chordal graph} is a graph in which all cycles have at least one chord. The \textbf{suspension} $\nabla G$ of a graph $G$ is obtained by adding a new vertex and connecting it to all other vertices. \\

Finally, for notational simplicity, $G = (V, E, w)$ will refer to the original input to the metric repair problem throughout this entire paper. Additionally, throughout the paper $k$ and {\rm OPT} will be the size of the optimal solution. We shall use them interchangeably depending on the context. Finally, $L+1$ will be the length of the longest broken cycle.

\subsection{Previous Results}

Fan, et al.~\cite{Raichel2018} and Gilbert and Jain~\cite{Gilbert2017} studied the special case of MR$(G, \Omega)$ where $G = K_n$ and $\Omega = \mathbb{R}_{\le 0}$ (decrease case), $\mathbb{R}_{\ge 0}$ (increase case), and $\mathbb{R}$ (general case). They present two different kinds of results; hardness results for the different versions of the problem and structural results about solutions to these problems. In particular, the major results from these previous works are:

\begin{thm} \cite{Raichel2018, Gilbert2017} \label{thm:specialDecrease} The problem MR$(K_n, \mathbb{R}_{\le 0})$ can be solved in cubic time. 
\end{thm}

\begin{thm}\cite{Raichel2018} \label{thm:specialNP} 
The problems MR$(K_n, \mathbb{R}_{\ge 0})$ and MR$(K_n, \mathbb{R})$ are NP-Complete (In fact APX-Hard) and permit $O(OPT^{1/3})$ approximation algorithms. 
\end{thm}

\begin{thm} \citep{Gilbert2017} \label{thm:oracle} 
Given a complete weighted graph $G$ and a support $S$ such that there exists an increase only solution on $S$ for $G$ and  $G - S$ is a connected graph, then for any edge $uv \in S$, setting the weight of $uv$ to be the shortest distance between $u$ and $v$ is a valid increase only solution. 
\end{thm}

\begin{thm} \cite{Raichel2018} \label{thm:specialStrctr} For a complete weighted graph $G = (V,E,w)$ with non negative weights and $S \subset E$ we have the following: 
\begin{enumerate}[nosep]
\item If $S$ contains an edge from each broken cycle, then $S$ is the support to solution to MR$(G, \mathbb{R})$. 
\item If $S$ contains a bottom edge from each broken cycle. Then $S$ is the support to solution to MR$(G, \mathbb{R}_{\ge 0})$.
\end{enumerate} 
\end{thm}

\subsection{Generalizations to Graph Metric Repair} 

As with previous work, we will focus on the cases when $\Omega = \mathbb{R}_{\le 0}, \, \mathbb{R}_{\ge 0}$, and $\mathbb{R}$. The goal of the next few subsections is to generalize theorems \ref{thm:specialDecrease}, \ref{thm:specialNP}, \ref{thm:specialStrctr}, and  \ref{thm:oracle} to the case when $G$ is any graph.

\subsubsection{Hardness Results}

As with the MVD and SMR problems the decrease version of the problem can still be solved in cubic time. A more detailed analysis of this and a few additional corollaries can be found in the appendix. For the increase and general cases it is a simple matter to extend Fan, et al.'s hardness results but we shall use the reduction to show a stronger result. Namely, if $\Omega$ contains at least one positive value, then the problem is NP-hard and inherently combinatorial in nature. 

\begin{thm} \label{thm:NP}  
If $0 \in \Omega$ and $\Omega \cap \mathbb{R}_{> 0} \neq \emptyset$ then we have that the problem MR$(G, \Omega)$  is NP Complete. 
\end{thm}
\begin{proof} Let $\alpha \in \Omega \cap \mathbb{R}_{> 0}$. Our goal is to take a general instance of vertex cover and reduce it to MR$(G, \Omega)$. Thus, given a graph $G$ set the weight of all edges in graph $G$ to $3\alpha$. Then, take the suspension $\nabla G$ of $G$ and make the weight of each new edge $\alpha$. At this point, this reduction is the same as presented in \cite{Raichel2018}, hence the proof of the reduction is the same and we omit it.
\end{proof}

\begin{cor} \label{cor:APX} 
If $0 \in \Omega$ and $\Omega \cap \mathbb{R}_{> 0} \neq \emptyset$, then the problem MR$(G, \Omega)$  is APX-hard, and assuming the Unique Games Conjecture, is hard to approximate within a factor of 2 $\epsilon$ for any $\epsilon > 0$. 
\end{cor}

While our reduction is essentially the same as that of \citet{Raichel2018}, we provide two new crucial insights. First, taking the suspension of a general $G$ is a natural structure to consider and provides a clean reduction from Vertex Cover to metric repair. Second, we can see that in the extreme case when $\Omega = \{0,\alpha\}$ for $\alpha > 0$ the problem is still NP-Hard and the difficulty comes from the combinatorial side of the problem.

\subsubsection{Structural Results}

Theorems \ref{thm:specialStrctr} and \ref{thm:NP} suggest that the problem is mostly combinatorial in nature. In general, we shall see that the difficult part of the problem is finding the support of of the solution. Given a solution $W$ to the Metric Repair problem, let us define $S_W = \text{supp}(W)$. We shall now present a characterization of the support of solutions to graph metric repair problem that generalizes Theorems \ref{thm:specialStrctr}, \ref{thm:oracle}. Here the key insight that lets us generalize the results is: \begin{enumerate} \item[(0)] If the shortest path between two adjacent vertices is the not the edge connecting them, then this edge is the top edge of a broken cycle \end{enumerate}

\begin{thm} \label{thm:structure} 
For any weighted graph $G = (V,E,w)$ with non negative weights and $S \subset E$, the following hold:
\begin{enumerate}[nosep]
\item $S$ contains an edge from each broken cycle if and only if $S$ is the support to solution to MR$(G, \mathbb{R})$. 
\item $S$ contains a bottom edge from each broken cycle if and only if $S$ is the support to solution to MR$(G, \mathbb{R}_{\ge 0})$.
\end{enumerate} 
\end{thm}
\begin{proof} A detailed proof is in the Appendix \end{proof}

Furthermore, given a weighted graph $G$ and a potential support $S_W$ for a solution $W$, in polynomial time (cubic in fact) we can determine whether there exists a valid (increase only or general) solution on that support and, if one exists, finds it. 

\begin{algorithm}[h]
\caption{Verifier}
\begin{algorithmic}[1]
\Function{Verifier}{G=(V,E,w),S}
\State $M = \|w\|_\infty$
\State For each $e \in S$ set $w(e) = M$
\State Set $w(v,u)$ to be length of the shortest path from $u$ to $v$ in the new graph
\If{Only edges in $S$ had their weight changed (or increased for increase only case) }
\State \Return w
\Else
\State \Return Not a valid support
\EndIf
\EndFunction
\end{algorithmic}
\end{algorithm}

\begin{prop} \label{prop:verifier} 
The {\sc Verifier} algorithm given a weighted graph $G$ and a potential support for a solution $S$ determine in polynomial (cubic in fact) whether there exists a valid (increase only or general) solution on that support and if one exists finds one. 
\end{prop}
\begin{proof} The key insight here is same as that for Theorem \ref{thm:structure}. The complete detailed proof can be found in the appendix.
\end{proof}

As the insight and theorems show, once we have found the support, the problem can be easily solved. In fact, as the next theorem says, once we know the support the set of all possible solutions on that support is a nice space. 

\begin{thm} \label{thm:compact} For any weighted graph $G$ and support $S$ then we have that the set of solution with support $S$ is a closed convex subset of $\mathbb{R}^{n \times n}$. Additionally, if $G - S$ is a connected graph or we require have an upper bound on the weight of each edge, then set of solutions is compact \end{thm}
\begin{proof} A detailed proof is in the Appendix \end{proof}

%!TEX root = ./main.tex
\section{Fixed parameter analysis}
\label{sec:fpt}

Because we reduce from Vertex Cover to show the difficulty of the Metric Repair problem, it might be the case that MR$(G, \Omega)$ is also fixed parameter tractable. \citet{Raichel2018} remarked that the MVD problem might not be fixed parameter tractable, which could explain the current gap in the approximation ratios for these problems. It turns out that not only is the MVD problem fixed parameter tractable but the generalized graph metric problem is also fixed parameter tractable when $G$ is any chordal graph, if we parameterize by the size of the optimal solution. 

\begin{rem} Chordal graphs appear in a variety of different places and several problems that are hard on general classes of graphs are easy on chordal graphs. For example, maximal clique and graph coloring can be solved in polynomial time on chordal graphs \cite{Stacho08}. Chordal graphs also appear in the realm of Euclidean distance matrix and positive semi-definite matrix completion \cite{Gower1985}. \end{rem}
\begin{rem} Finding maximal chordal subgraphs and minimal chordal supergraphs is NP-hard \cite{YM1, YM2} \end{rem}

\subsection{Increase Only Case}

We shall start by focusing on the increase version of the problem. There are four major insights that let us develop a FPT algorithm for MR($G, \mathbb{R}_{\ge 0}$). These insights are crucial as they reduce the search space and let us \emph{recursively build the support while simultaneously expanding the search space} \begin{enumerate}

\item[(1)] The {\sc Verifier} tells us given a support we can determine if it is a valid support in polynomial time. See Proposition \ref{prop:verifier}.

\end{enumerate}
\noindent Now the naive approach would be to look at all subsets of edges of size $k$ and check if they are valid supports. But there are $\Theta(n^{2k})$ many such subsets so a brute force search is not a valid FPT algorithm. Thus, we want to reduce the search space. 
\begin{enumerate}
\setcounter{enumi}{1}

\item[(2)] If an edge $e$ is a bottom edge in more than $k$ broken 3-cycles then $e$ must be in the support of all optimal solutions.  See Lemma \ref{lem:opt}.

\end{enumerate}

\noindent The next two insights then let us recursively build the support while simultaneously expanding the search space. 

\begin{enumerate}
\setcounter{enumi}{2}

\item[(3)] In a chordal graph the presence of a broken cycle implies the presence of a broken 3-cycle. This allows us to recursively pick edges for the support until we have no more broken 3-cycles. See Lemma \ref{lem:chordal}.

\item[(4)] If $e$ is an edge in the support of an optimal solution $W$ and we modify $w(e) \leftarrow W(e)+w(e)$ and make no other changes, then $e$ cannot be a top edge in more than $k$ broken 3-cycles. See Lemma \ref{lem:recurse}.
\end{enumerate}

\begin{algorithm}[h]
\caption{FPT algorithm Increase Metric Repair (FPIMR)}
\label{alg:fptInc}
\begin{algorithmic}[1]
\Function{FPIMR}{G=(w,V),k}
	\State  $S = \emptyset$
	\State $T = $ Set of all broken 3-cycles in $G$
	\State For all edges $e$ in more than $k$ broken triangles as a bottom edge add $e$ to $S$
	\State $P = \emptyset$
	\State For each triangle in $T$ whose bottom edges are not in $S$, add both bottom edges to $P$
	\For{each $e = (ij) \in S$}		\State Sort $\{|w_{il}-w_{jl}| l = 1,\hdots, n\}$
		\State For the $k$ biggest entries add edges add the edge with bigger weight from $il, lj$ to $P$
	\EndFor
	\State S = \Call{Cover}{$(G,S,P,k)$}
	\State \Return{$S$}
\EndFunction

\end{algorithmic}
\end{algorithm}

\begin{algorithm}[h]
\caption{Cover for FPT algorithm Increase Only Case}
\label{alg:coverInc}
\begin{algorithmic}[1]
\Function{Cover}{G=(w,V),S,P,k}
	\If{k=|S|}
		\State \Return \Call{verifier}{G,S}
	\EndIf
	\For{new $e = (ij) \in S$}  \Comment{Implemented by storing 1 bit extra per edge in $S$ as an indicator}

		\State Sort $\{w_{il}+w_{lj} | l = 1, \hdots, n\}$
		\State For the smallest $k$ entries add edge $il, jl$ to $P$
	\EndFor
	\For{$e \in P$}
		\State S = \Call{Cover}{($G,S \cup \{e\}, P-\{e\}, k$)}
		\If{$S \neq$ NULL}
			\State \Return $S$
		\EndIf
	\EndFor
	\State \Return NULL
\EndFunction
\end{algorithmic}
\end{algorithm}

The general structure of the algorithm will be as follows. We shall have two sets $S$ which will contain the edges that we are currently considering as the \textit{support} and $P$ which will contain the edges that we could \textit{potentially} add to $S$. We will start by adding edges that have to be in the support of all optimal solutions to $S$. We then want to recursively keep adding edges to $S$ from $P$ while simultaneously also expanding $P$. We prove our results in a series of lemmas that correspond to our insights. The detailed proofs for these lemmas can be found in the appendix. 

\begin{lem} \label{lem:opt} The edges added to $S$ on line 4 of {\sc FPIMR} are in the support of all optimal solutions. \end{lem}

\begin{lem} \label{lem:chordal} If $G$ is a chordal weighted graph that has a broken cycle then $G$ has a broken triangle. \end{lem}

\begin{lem} \label{lem:recurse} If $W$ is an optimal solution and $S \subsetneq S_W$, then $\exists\, e \in S_W \cap P$ such that $e \not \in S$. \end{lem}

\begin{lem} \label{lem:bound} At all stages of the {\sc FPIMR} and {\sc Cover} algorithms we have that $|P| \le 5k^2$. \end{lem}

\begin{thm} \label{thm:FPT1} If we restrict $G$ to being a chordal graph. Then MR$(G, \mathbb{R})$ is fixed parameter tractable when parameterized by the size of the optimal solution and can be solved in $O((12k^2)^kn^3)$ time by the {\sc FPIMR} algorithm. \end{thm}

\begin{proof}We have already proved all the major pieces we just need to put them together.  Lemma \ref{lem:opt} tells us that when we initially call {\sc Cover} we have that $S \subset S_W$ for any optimal solution $W$. Lemma \ref{lem:recurse} then tells us we can continue recursively until we have found the support of an optimal solution. Finally, Proposition \ref{prop:verifier}, tells us that the {
\sc Verifier} finds an optimal solution.

Now we need to show that the algorithm runs in $O((5k^2)^kn^3)$ time. Lemma \ref{lem:bound} tells us that $|P| \le 5k^2$ at all times. Thus, at each recursive stage we have a branch factor of at most $5k^2$. Now we also know that our recursive depth is at most $k$. Thus we have at most $(5k^2)^k$ nodes in our recursion tree. At each non terminal non root node we do $O(n\log(n))$ work, at the root node of the tree  we do $O(kn\log(n))$ work, and at each terminal node we do $O(n^3)$ work. Thus, we see that the {\sc Cover} algorithm runs in $O((5k^2)^kn^3 + kn\log(n)) =O((5k^2)^kn^3)$. We also do some work before we start the recursive procedure but this can clearly be done in $O((5k^2)^kn^3)$ time.  
\end{proof}

\subsection{General Problem is FPT For Chordal Graphs}

Not only is the increase only problem FPT, but also the general problem. To analyze this more general problem, we modify insights (2) and (4) slightly. 

\begin{enumerate}
\setcounter{enumi}{1} 
\item[(2')] If any edge $e$ is in more than $k$ broken cycles as a \emph{top edge or as bottom edge}. Then it must be in all optimal solutions. 
\setcounter{enumi}{3}
\item[(4')] If $e$ is an edge in an optimal solution $W$ and we modify $w(e) \leftarrow W(e)+w(e)$ and make no other changes, then $e$ cannot be the \emph{top edge in more than $k$ broken triangles} and $e$ cannot be a \emph{bottom edge in more than $k$ broken triangles}. 
\end{enumerate}

A complete analysis, including the algorithm and detailed proofs for the general case can be found in the appendix. In particular we get the following result:

\begin{thm}  \label{thm:FPT2} If we restrict $G$ to being a chordal graph. Then MR$(G, \mathbb{R})$ is fixed parameter tractable when parameterized by the size of the optimal solution and can be solved in $O((12k^2)^kn^3)$ time. \end{thm}
%\input{approx}
%\input{generalization}
%!TEX root = ./main.tex
\section{Approximation analysis}
\label{sec:approx}

\citet{Raichel2018} and \citet{Gilbert2017} both present algorithms for the special case of the graph metric repair problem when $G =K_n$ and $\Omega =\mathbb{R}_{\ge 0}$. These algorithms, however, depend on the graph $G$ being a complete graph and can not be used for the graph metric repair problem. In this section, we will present an $L$-approximation algorithm for the increase case and an $L+1$-approximation algorithm for the general case that both run in $O(n^3({\rm OPT}+1))$ time. We first give a comparative analysis of the approximation algorithms for this special case. We also show that with a slight modification we can reduce the running time of the algorithm from \cite{Raichel2018} to $\Theta(n^5)$ with the same approximation ratio. Additionally, we shall show that {\sc IOMR-fixed} from \cite{Gilbert2017} is an $\Omega(n)$ approximation algorithm, disproving their conjecture that {\sc IOMR-fixed} is a 2 approximation algorithm.

\subsection{Comparison of approximation algorithms for $G = K_n$ and $\Omega = \mathbb{R}_{\ge 0}$}

Table \ref{table:alg} summarizes all known approximation algorithms for MR$(K_n, \mathbb{R}_{\ge 0})$.

\begin{table*}[h]
\centering
\begin{tabular}{|c|c|c|c|} \hline
Algorithm & Running time & Approximation Ratio & Reference \\ \hline \hline 
 & $\Theta(n^6)$ & $O({\rm OPT}^{1/3})$ & \citet{Raichel2018} \\ \hline
{\sc SPC} & $O(n^3( {\rm OPT}+1))$ & $L$ & This paper \\ \hline
{\sc 5-Cycle Cover} & $\Theta(n^5)$ & $O({\rm OPT}^{1/3})$ & This paper \\ \hline
{\sc IOMR-fixed} & $\Theta(n^3)$ & $\Omega(n)$ & \citet{Gilbert2017} \\ \hline 
\end{tabular}
\caption{For $G=K_n$, $\Omega = \mathbb{R}_{\ge 0}$, running times and approximation ratios for all known algorithms.}
\label{table:alg}
\end{table*}

To compare these algorithms let us start by first considering the length of the longest broken cycle. Suppose the length of the longest broken cycle is $o({\rm OPT}^{1/3})$, then {\sc SPC} gives the best approximation ratio. While if we have many broken cycles of length $\omega({\rm OPT}^{1/3})$, then {\sc 5-Cycle Cover} gives us the best approximation ratio. 

To see that both cases are possible consider the following example. Let the length of the biggest broken cycle be $L$ and divide the $n$ vertices into $n/L$ components of size $L$. For each component, we pick $L/2$ vertex disjoint edges and let them weight 1, the rest of the edges in the component have weight 0, and the edges between various components have weight 2. Then, any broken cycle must be contained in one component, the size of the optimal solution is $\Theta(nL)$, and the length of the largest broken cycle is $\Theta(L)$. For these graphs we see that if $L = o({\sqrt{n}})$, then $L = o({\rm OPT}^{1/3})$. On the other hand if if $L = \omega({\sqrt{n}})$, then $L = \omega({\rm OPT}^{1/3})$.

On the other hand, if we have dense optimal solutions, that is, $\min(L\cdot {\rm OPT}, {\rm OPT}^{4/3}) = \Theta(n^2)$, then we have that all three solutions fix $O(n^2)$ entries. In this case, since {\sc IOMR-Fixed} has a much faster run time than {\sc Short Cycle Cover} and {\sc 5-Cycle Cover}, {\sc IOMR-Fixed} is best option.

\subsection{{\sc Short Path Cover} Approximation Algorithhm}

The main idea for these algorithms is that: \begin{enumerate} \item[(6)] We can approximate the set of all broken cycles by iteratively solving the APSP problem on a decreasing chain of graphs until we obtain a graph with no broken cycles.\end{enumerate} Specifically, once we find all shortest paths in a graph, insight (0) tells that each of these paths defines a broken cycle. Thus, we cover all these broken cycles, remove these edges from the graph, and then repeat. In general we can cover the broken cycles we have found in a variety of different ways, with each one giving us a different approximation algorithm. In this paper we shall present one particular way of doing this in {\sc Short Path Cover (SPC)}.

\begin{algorithm}
\caption{Short Path Cover (SPC) for the Increase Case of Graph Metric Repair}
\begin{algorithmic}[1]
\Function{SPC}{G}
\State $S = \emptyset$
\State $P$ = Shortest Paths between all adjacent vertices
\State Remove all paths that are a single edge
\While{$P$ is not empty}
\While{$P$ is not empty}
\State Let $p$ be a path in $P$. Remove all edges in $p$ from $G$ and add them to $S$
\State Remove all paths from $P$ that intersect $p$
\EndWhile
\State $P$ = Shortest Paths between all adjacent vertices in new graph
\State Remove all paths that are a single edge
\EndWhile
\State \Return \Call{Verifier}{G,S}
\EndFunction
\end{algorithmic}
\end{algorithm}

\begin{thm} \label{thm:spc} 
{\sc SPC} is an $L$-approximation algorithm for the increase only problem that runs in $O(n^3(k+1))$ time. 
\end{thm}
\begin{proof} To see that this is a valid solution, we need to show that $S$ has at least one edge from each broken cycle. Let $C$ be a broken cycle in $G$ with top edge $e$. Now if $e \not\in S$. Then the we know that since $G - S$ has no broken cycles, at least one bottom edge of $C$ must be in $S$. On the other hand if $e \in S$, let $\tilde{S}$ be the support found by {\sc SPC} just before it adds $e$ to $S$. Then when we add $e$ to $S$, it must be the case that $e$ is an edge of the shortest path between two vertices. Hence, at this stage the edge $e$ is the shortest path between its end points. Thus, thus the cycle $C$ is not present in $G - \tilde{S}$. Hence, a bottom edge from $C$ is in $\tilde{S} \subset S$. Thus, $S$ has at least one bottom edge from each broken cycle. Then by theorem \ref{thm:structure}, $S$ is a valid support. 

When we find an uncovered broken cycle we know that at least one bottom edge must be in the optimal solution. However we add all bottom edges. Thus, in the worst case we add $L$ times the number of needed edges.

Next we show that we terminate after $O(k)$ iterations of the outer loop. If we have not terminated, then we must have that $P$ is not empty. Now let $p \in P$ be first path consider. Now $p$ along with the edge between its end points is an uncovered broken cycle. Hence, at least one bottom edge $e$ from this cycle must be in the optimal solution. Then, since we add all bottom edges to the support, we will $e$ to $S$ as well Thus, we see that after at most $k$ iteration of the outer loop we must have covered all broken cycles.  Finally,  each iteration of the outer loop runs in $O(n^3)$ time. Thus, whole algorithm runs in $O(n^3(k+1))$ time.
\end{proof}

\begin{prop} \label{prop:tight} The approximation ratio is tight for {\sc SPC}. 
\end{prop}
\begin{proof} Consider $G = C_n$ where one edge has weight $n$ and then rest of the edges have weight one. 
\end{proof}

\begin{rem} \label{rem:dense} The above proposition shows that the approximation ratio is tight but for extremely sparse graphs. When we have more dense graphs we believe that this algorithm does better as for denser graphs the short paths most likely have more than one edge in the optimal solution. For example on the first iteration let us look at all length two segments  of a shortest path $P$ (see Figure \ref{fig:dense}). Then since $P$ is a shortest path the triangles spanned by any length two segment must be broken. Thus, we actually need to pick at least half of the edges in this path instead of just one. Additionally, for dense graphs we could potentially keep finding small cycles and never encounter a broken cycle of large length.  
\end{rem}

\begin{figure}[h]
\centering
\includegraphics[width = 0.25\textwidth]{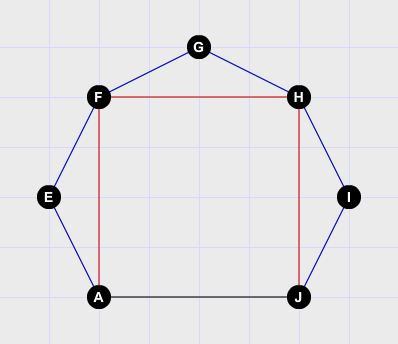}
\caption{Example for remark 3. The black edge is the top edge and the blue edges are the bottom edges.}
\label{fig:dense}
\end{figure}

\noindent\textbf{General Case.} For the general problem we modify {\sc SPC} as follows. Instead of adding all the bottom edges when we consider a cycle, we add all the edges in the cycle. The rest of the analysis is similar to before and can be found in the appendix. 

\begin{thm} \label{thm:AppxGen} 
{\sc General SPC} is an $L+1$-approximation algorithm for the MR$(G, \mathbb{R})$ problem that runs in $O(n^3(k+1))$ time. 
\end{thm}
%\begin{proof} It is clear that the set $S$ that {\sc General SPC} finds covers all broken cycles. Thus, {\sc Verifier} finds a valid solution to the MR$(G, \mathbb{R})$ problem with support $S$. 
%\end{proof}

\subsection{5 Cycle Cover and IOMR-fixed}

For {\sc 5 Cycle Cover} it is sufficient to cover all cycles of length $\le 5$ instead of cycles of length $\le 6$ as presented in \cite{Raichel2018}. Then, with no other modification to the rest of the algorithm we obtain an \textbf{$O({\rm OPT}^{1/3})$ that runs in $\Theta(n^5)$ instead of $\Theta(n^6)$}. For a more detailed treatment of this algorithm refer to the appendix. 

Gilbert and Jain presented {\sc IOMR-fixed} with no theoretical analysis. It is an appealingly simple algorithm with good empirical performance. It's also quite efficient and, with our theoretical analysis, is a good algorithm in certain settings. In particular we show that 
\textbf{{\sc IOMR-fixed} is an $\Omega(n)$ approximation algorithm}. The details for the results and some more analysis can be found in the appendix.

%\input{Generalization}
%!TEX root = ./main.tex
\section{Future Work}

We showed that if the underlying $G$ is a chordal graph, then generalized metric repair is fixed parameter tractable. We believe that it should be possible to find an FPT-reduction from Set Cover to generalized metric repair, thus showing that the problem in general is at least W[2] hard. (The authors at the time of writing posit it is actually harder.) This dichotomy would have several interesting implications. 

The first implication is that if the problem is at least W[2] hard, then there most likely exists a family of graphs $\mathcal{F}$ such that the generalized metric repair problem restricted to $\mathcal{F}$ is W[1] hard. In general, if it is harder, this would imply the existence of an increasing chain of families of graphs that all lie on different levels of the W hierarchy. 

A second interesting implication is that perhaps some other parameter besides the size of the optimal solution characterizes differences in hardness amongst various problem instances. For example, if we parameterize by the length of the longest broken cycle, then perhaps there is an interesting split between sparse and dense graphs. In particular, if we have only a linear number of edges, then we can find all broken cycles in fixed parameter time (which is not something we can do in general). 

Another avenue of future work is approximation algorithms. There may be a 2-approximation algorithm for chordal graphs or, in general, better approximations than $O(\log(n))$ for the general problem. We currently do not know if a $O(\log(n))$ approximation algorithm exists. In addition, the current known approximation algorithms are slow and we would like more efficient ones for data analysis applications. One large goal is to apply the metric repair algorithms to such machine learning problems as metric learning, dimension reduction, and data imputation.

One final avenue of future work was introduced in the introduction of this paper. We gave an example of how the TSP could be solved using approximations algorithms for the generalized graph metric repair problem. Hence, it would be interesting to evaluate the performance of this method and the quality of the solutions to the TSP thus obtained.

\nocite{*}
\bibliographystyle{plainnat}
\bibliography{citations} 

\newpage
%!TEX root = ./main.tex
\section{Appendix}

\subsection{Decrease only is in P}

Let us start with the decrease only case ($\Omega = \mathbb{R}_{\le 0}$). As with $G = K_n$, the general decrease only case is easy and can be solved in cubic time by the {\sc Dmr} algorithm. 

\begin{algorithm}
\label{alg:DMR}
\caption{Decrease Metric Repair ({\sc Dmr})}
\begin{algorithmic}[1]
\Function{DMR}{G = (V,E,w)}
\State Let $W = w$
\State For any edge $uv \in E$, set $W(e)$ =  weight of shortest paths between $u$ and $v$ 
\State \Return $w-W$
\EndFunction
\end{algorithmic}
\end{algorithm}

\begin{thm} \label{thm:decrease} MR$(G,\mathbb{R}_{\le 0})$ can be solved in $O(n^3)$ time by the {\sc Dmr} algorithm.\end{thm} 

\begin{proof} Let $e \in G$ be an edge whose edge weight is bigger than the shortest path between the two end points of $e$. Then in this case $e$ is the top edge in a broken cycle. Hence, any decrease only solution must decrease this edge. Thus all edges decreased by {\sc Dmr} are edges that must be decreased.

By the same reasoning we see that this new weighted graph has no broken cycles. Thus, we see that our algorithm gives a sparsest solution to MR($G, \mathbb{R}_{\le 0})$ in $\Theta(n^3)$ time. 
\end{proof}

\begin{cor} \label{dor:decp} 
For any $G = (V,E,w)$ {\sc Dmr} returns the smallest solution for any $\ell_p$ norm for $p \in [1, \infty)$. 
\end{cor}
\begin{proof} The proof of Theorem \ref{thm:decrease}  actually shows that there is a unique support for the sparsest solution, in fact any decrease only solution must contain these edges in its support. We can also see that {\sc Dmr} decreases these by the minimum amount so that the cycles are not broken. Thus, this solution is in fact the smallest for any $\ell_p$ norm. 
\end{proof}

\subsection{Structural Results}
{\reftheorem{thm:structure} 
For any weighted graph $G = (V,E,w)$ with non negative weights and $S \subset E$, the following hold:
\begin{enumerate}
\item $S$ contains an edge from each broken cycle if and only if $S$ is the support to solution to MR$(G, \mathbb{R})$. 
\item $S$ contains a bottom edge from each broken cycle if and only if $S$ is the support to solution to MR$(G, \mathbb{R}_{\ge 0})$.
\end{enumerate} }
\begin{proof} The only if part of both statements are straightforward to show. For (1) suppose that $S$ does not have an edge from all broken cycles. Let $C$ be such a broken cycle. Then no matter what changes we make to the weights of the edges in $S$, $C$ will remain a broken cycle. Hence, contradicting $S$ being the support of a valid solution. For (2) suppose that $S$ does not have a bottom edge from each broken cycle and let $C$ be such a  broken cycle. Then, no matter how we increase the weights $C$ will remain a broken cycle. Thus, $S$ must have had a bottom edge from all broken cycles. 

For the if part of both statements (1) and (2), we have the same crucial insight, that is, if the shortest path between two adjacent vertices is the not the edge connecting them, then this edge is the top edge of a broken cycle. 

Thus, for both statements let us define the following graphs, $\hat{G} = (V,E-S,w)$. Since in both cases we have that $S$ covered all broken cycle, we see that when we remove the edges in $S$, the new graph has no broken cycles. Thus, the shortest path between all adjacent vertices in $\hat{G}$ is the edge connecting them.

Then for each $e \in S$, set $w(e)$ be the length of the shortest path between its end points in $\hat{G}$ or $\|w\|_\infty$ if no path exists. First let us see that the new weights adhere to a metric. If the edge weight is shortest path between its nodes then it is not a top edge. Hence, edges in $\hat{G}$ and edges whose weight was set to length of the shortest path between its end points in $\hat{G}$.  Thus, we only need to look edges whose weight was set to $\|w\|_\infty$, then it connected two disconnected components in $\hat{G}$. Thus, any cycle with such an edge much involve another edge between components which also has weight $\|w\|_\infty$. Thus, these are not top edges. Thus, there are no top edges. Hence, the metric has been repaired. For the case MR$(G, \mathbb{R})$ we are done. 

For the caseMR$(G, \mathbb{R}_{\ge 0})$, suppose the weight of some edge $e \in S$ was decreased. This implies that $\hat{G}$ along with $e$ has broken cycle for which $e$ is the top edge. Let $P$ be the path between the two edge points of $e$ in $\hat{G}$. Now we assumed that $S$ had a bottom edge from each broken cycle. So it must have an edge from this broken cycle, but then this path could not have existed in $\hat{G}$. Thus, we see that we set $w(e) = \|w\|_\infty$. Thus, it could not have been decreased and we get an increase only solution. 
\end{proof}

{\reftheorem{thm:compact} For any weighted graph $G$ and support $S$ then we have that the set of solution with support $S$ is a closed convex subset of $\mathbb{R}^{n \times n}$. Additionally, if $G - S$ is a connected graph or we require have an upper bound on the weight of each edge, then set of solutions is compact.}
\begin{proof} Let $x_{ij}$ for $1 \le i,j \le n$ be our coordinates. Then the equations $x_{ij} = c_{ij}$ for $ij$ not in the support and $x_{ij} \le x_{ik}+x_{kj}$ define a closed convex set. Thus, we see the first part. For the second we just need to see that set is bounded to get compactness. If we have thatt $G - S$ is connected then for all $e \in S$ there is a path between end points of $e$ in $G-S$. Thus, the weight of this path is an upper bound. On the other hand 0 is always a lower bound. Thus, we get compactness if $G-S$ is connected \end{proof}

{\refprop{prop:verifier}  The {\sc Verifier} algorithm given a weighted graph $G$ and a potential support for a solution $S$ determine in polynomial (cubic in fact) whether there exists a valid (increase only or general) solution on that support and if one exists finds one. }

\begin{proof} First, let us see that if edges not in $S$ are changed then this is not a valid support. Let $v_1v_2$ be the edge that is changed. Then, there exists a path $v_1 = u_1, \hdots, u_k = v_2$ that has shorter total weight than $v_1v_2$ and this forms a broken cycle with $v_1v_2$ as the top edge. 

Since the weight of all edges in $S$ are first changed to the maximum weight in the graph, no bottom edge in the cycle can be in $S$. We also know its not the top edge. Thus, no matter what changes we make to the weights of the edges in $S$ this broken cycle will persist. 

For the increase only case, we have to also make sure that the solution returned doesn't decrease any edge weights in $S$. Let us assume that $S$ is the support of a valid increase only solution and we decreased an edge $e$ in $S$. Then, using similar logic as before, we see that $e$ is the top edge in a broken cycle $C$ and none of the bottom edges belong to $S$. But now the total weight of the bottom edges in the cycle $C$ is smaller than the original edge weight of $e$. Thus, $C$ was initially broken cycle for which we have no bottom edges in $S$. Thus, by Theorem \ref{thm:structure} there couldn't be any increase only solution with this support.
\end{proof}

\subsection{FPT}

{\reflemma{lem:opt} All edges added to $S$ on step 4 of algorithm \ref{alg:fptInc} are in all optimal solutions.}
\begin{proof} As we mentioned this is just insight (2). This is because if an edge $uv$ is in more than $k$ broken triangles. Then besides the edge $uv$ these triangles are all edge disjoint. Thus we would need at least 1 edge from each triangle in the optimal solution. Which contradicts the fact that the optimal solution had size $k$ \end{proof}

{\reflemma{lem:chordal} If $G$ is a chordal weighted graph that has a broken cycle then $G$ has a broken triangle.}
\begin{proof} This is insight (3). See appendix. Let $v_1, \hdots, v_k$ be the smallest (in terms of number of vertices) broken cycle in $G$, with $v_1v_k$ being the top edge. Now there are two possibilities, either $k = 3$ in which we are done otherwise since $G$ is chordal, this cycle has at least 1 chord present in the graph. Let $v_iv_j$ be this chord. Now there are two paths from $v_i$ to $v_j$ on the cycle. Let $P_1, P_2$ be these two paths such that $P_1$ contains the top edge. Now we have two cases. Since $k$ was the length of the smallest broken cycle we must have that \[ w(v_i, v_j) \le  \sum_{e \in P_2} w(e) \] because otherwise this would have been a smaller broken cycle. In this case $P_1$ along with $v_iv_j$ is a broken cycle with $v_1v_k$ as the top edge (since the path using the edge $v_iv_j$ is shorter than the weight of $P_2$). But now this is a smaller broken cycle. Hence, we have a contradiction again. Thus we see that the smallest broken cycle must be a 3 cycle. Thus we have proven the claim.

\end{proof}

For notational convenience given an optimal solution $W$ for a weighted graph $G$ and $A \subset S_W$, let $G_A = (V,E,w_A)$ be the weighted graph in which all $e \in A$, we set $w_A(e) = W(e)+w(e)$ (i.e. the graph in which we have fixed this edges in $A$)

{\reflemma{lem:recurse} If $W$ is an optimal solution and $S \subsetneq S_W$, then there exists an edge in $e \in S_W \cap P$ such that $e \not \in S$.}
\begin{proof} This uses insights (2) and (4). Since we have that $S$ is a proper subset of $S_W$ we know that $S$ cannot be the support of a solution. Then, by Lemma \ref{lem:chordal}, we know that $G_S$ has a broken triangle $T = xyz$, with $xy$ as the top edge. Now if both of the bottom edges $xz,yz$ are  in $S$, then we know these edges have been set to their correct weights (according to $W$) but we still have a broken triangle. Thus, at least one of these edges is not in $S$. There are 3 cases we have to consider:

\textbf{Case 1: No edge of $T$ is in $S$}. In this case $T$ is broken in the original graph $G$ as well and does not have any edge from $S$. Thus, we would have added both bottom edges to $P$ initially on line 6 of algorithm \ref{alg:fptInc} and the claim holds. \\

\textbf{Case 2: The top edge $e = xy$ of $T$ is in $S$}. Let us consider $G_{\{e\}}$. We first note that since we only increase weights $T$ is a broken triangle in $G_{\{e\}}$. Now we know that the size of the optimal solution for $G_{\{e\}}$ is $k - 1$. Hence in $G_{\{e\}}$, $e$ cannot be the top edge in more than $k$ triangles, because we would then need $k$ many new edges. Thus, we add both bottom edges in $T$ to $P$ on lines 7-8 of algorithm \ref{alg:coverInc}. If we had not then $w(x,z)+w(z,y)$ was not among the $k$ smallest sums and, hence, $e$ would have been a top edge in more than $k$ broken triangles. Which is a contradiction and, again, our claim holds. \\

\textbf{Case 3: A bottom edge $e = yz$ in $S$ and the top edge $xy$ in $T$ is not in $S$}. Since we have only increased the bottom edges of $T$, it was broken in the original graph $G$, and we added the other bottom edge to $P$ on line 8-9 of algorithm \ref{alg:fptInc}. 

To see why this is the case, let us suppose we had not and let $e_t, e_b$ be the top and the other bottom edge of $T$. Then, $|w(e_t) - w(e_b)|$ is less than $k$ larger values. In which case, $e$ is in more than $k$ broken triangles in $G_{\{e\}}$. But this is a contradiction as these triangles are edge disjoint besides the edge $e$. Hence, we would have included $e_b$ in $P$ on line 8-9 of algorithm \ref{alg:fptInc}, and, again, the claim holds \\

Thus, we see that in all cases we have at least one edge in $P \cap S_W$ that is not in $S$ and the claim holds.  \end{proof}

{\reflemma{lem:bound} At all stages of the algorithm we have that $|P| \le 5k^2$.}
\begin{proof} Let us first look at the number of edges added to $P$ by algorithm \ref{alg:fptInc}. We add edges to $P$ twice in algorithm \ref{alg:fptInc}. The first time we add edges, we look at all broken triangles that do not have any edges in $S$. I claim that there are at most $k^2$ such broken triangles. This is because we know all broken triangles can be covered by $k$ edges. But since none of these triangles have either bottom edge in $S$, we see that each edge in the solution is in at most $k$ of these triangles. Thus we have at most $k^2$ broken triangles. Then we add both bottom edges to $P$. Thus we have added $2k^2$ edges to $P$. The second time we add edges to $P$ in algorithm \ref{alg:fptInc} we look at each edge $e$ in $S$ and add at most $k$ edges for each of these edges. Thus since $|S| \le k$ we have that again we add at most $k^2$ edges to $P$

Let us look at the edges we add to $P$ in algorithm \ref{alg:coverInc}. We see that for any edge in $S$ we add at most $2k$ edges to $P$ once. Since $|S| \le k$ at all times we see that we have added at most $2k^2$ edges to $P$. Thus in total we see that at all times we have that $|P| \le 5k^2$. \end{proof}

\subsection{General FPT}

\begin{algorithm}[h]
\caption{FPT algorithm General only Case}
\label{alg:fpt}
\begin{algorithmic}[1]
\Function{fpt}{G=(w,V),k}
	\State  $S = \emptyset$
	\State $T = $ Set of all broken 3-cycles in $G$
	\State For all edges $e$ in more than $k$ broken triangles add $e$ to $S$
	\State $P = \emptyset$
	\State For each triangle in $T$. If it has no edges from $S$ add all edges to $P$
	\For{Each $e = (ij) \in S$}
		\State Sort $\{|w_{il}-w_{jl}| l = 1,\hdots, n\}$
		\State For the $k$ biggest entries add edges $il, lj$ to $P$
		\State Sort $\{|w_{il}+w_{jl}| l = 1,\hdots, n\}$
		\State For the $k$ smallest entries add edges $il, lj$ to $P$
	\EndFor
	\State S = \Call{Cover}{$(G,S,P,k)$}
	\State \Return{$S$}
\EndFunction
\end{algorithmic}
\end{algorithm}

\begin{algorithm}[h]
\caption{Cover for FPT algorithm General Only Case}
\label{alg:cover}
\begin{algorithmic}[1]
\Function{Cover}{G=(w,V),S,P,k}
	\If{k=|S|}
		\State \Return \Call{verifier}{G,S}
	\ElsIf{|S| > k}
		\State \Return NULL
	\EndIf
	\For{$e = (ij) \in S$}
		\State Sort $\{w_{il}+w_{lj} | l = 1, \hdots, n\}$
		\State For the smallest $k$ entries add edge $il, jl$ to $P$
		\State Sort $\{w_{il}-w_{lj} | l = 1, \hdots, n\}$
		\State For the biggest $k$ entries add edge $il, jl$ to $P$
		\State For $(i,l) \in S$, add $(jl)$ to $S$
		\State For $(j,l) \in S$, add $(il)$ to $S$ 
	\EndFor
	\For{$e \in P$}
		\State S = \Call{Cover}{($G,S \cup \{e\}, P-\{e\}, k$)}
		\If{$S \neq$ NULL}
			\State \Return $S$
		\EndIf
	\EndFor
	\Return S
\EndFunction
\end{algorithmic}
\end{algorithm}

Now the relevant lemmas are as follows.

\begin{lem} \label{lem:optg} All edges added to $S$ on step 4 are in all optimal solutions \end{lem}
\begin{proof} As we mentioned this is just insight (2). This is because if an edge $uv$ is in more than $k$ broken triangles. Then besides the edge $uv$ these triangles are all edge disjoint. Thus we would need at least 1 edge from each triangle in the optimal solution. Which contradicts the fact that the optimal solution had size $k$ \end{proof}

\begin{lem} \label{lem:recurseg} If $W$ is an optimal solution and $S \subsetneq S_W$, then there exists an edge in $e \in S_W \cap P$ such that $e \not \in S$ \end{lem}

\begin{proof} Again Lemma \ref{lem:chordal} tells us that $G_S$ has a broken triangle $T = xyz$ with $xy$ being the top edge. Again we see that since the triangle is broken one of its 3 edges cannot be in $S$. Thus, we have 4 cases.

\textbf{Case 1: No edges $T$ in $S$}. Then we added all 3 edges to $P$ on line 6 of algorithm \ref{alg:fpt} and hence, we the claim holds. \\

\textbf{Case 2: Two of the three edges from $T$ in $S$}. Then lines 11,12 of algorithm \ref{alg:cover} add the the third edge to $P$ and hence, the claim holds. \\

\textbf{Case 3: One edge from $T$ in $S$ and the edge was increased}. We added the requisite edges to $P$ on line 9 of algorithm \ref{alg:fpt} and line 8 of algorithm \ref{alg:cover}. The reasoning is the same as in case 2 and case 3 of Lemma \ref{lem:recurse} and hence, the claim holds. \\

\textbf{Case 3: One edge from $T$ in $S$ and the edge was decreased}.  We added the requisite edges to $P$ on line 11 of algorithm \ref{alg:fpt} and line 10 of algorithm \ref{alg:cover}. The reasoning is analogous to the ones in case 2 and case 3 of Lemma \ref{lem:recurse} and hence, the claim holds. \\

Thus, we see that in all cases we have at least one edge in $P \cap S_W$ that is not in $S$ and the claim holds. \end{proof}

\begin{lem} \label{lem:bounbg} At all stages of the algorithm we have that $|P| \le 12k^2$ \end{lem}

\begin{proof} Let us first look at the number of edges added to $P$ by algorithm  \ref{alg:fpt}.. We add edges to $P$ thrice in algorithm 5. The first time we look at all triangles that do not have any edges in $S$. I claim that there are at most such broken $k^2$ triangles. This is because we know all broken triangles can be covered by $k$ edges. But since none of these triangles have any edges in $S$, we see that each edge in the solution is in at most $k$ of these triangles. Thus we have at most $k^2$ of these triangles. 

Then we add all edges to $P$. Thus, we have added at most $3k^2$ edges to $P$. The second and third time we add edges to $P$ in algorithm we look at each edge $e$ in $S$ and each time add at most $2k$ edges for each of these edges. Thus, since $|S| \le k$ we have that again we add at most $4k^2$ edges to $P$

Now let us look at the edges we add to $P$ in algorithm  \ref{alg:cover}.. Now we see that for any edge in $S$ we add at most $4k$ edges to $P$ once. Now since $|S| \le k$ at all times we see that we have added at most $4k^2$ edges to $P$. 

Additionally if two edges in $S$ are adjacent, we add in the third edge (that would make it a triangle) to $P$. There are most $k^2$ such edges (since $|S| \le k$)

Thus in total we see that at all times we have that $|P| \le 12k^2$. \end{proof}

Hence now we can put these lemmas together in the same manner as before to get the following major theorem 

{\reftheorem{thm:FPT2} If we restrict $G$ to being a chordal graph. Then MR$(G, \mathbb{R})$ is fixed parameter tractable when parameterized by the size of the optimal solution and can be solved in $O((12k^2)^kn^3)$ time. }

\subsection{Approximation Algorithm}

\begin{algorithm}
\caption{Short Path Cover for General Graph Metric Repair}
\begin{algorithmic}[1]
\Function{General Short Path Cover}{G}
\State $S = \emptyset$
\State $P$ = Shortest Paths between all adjacent vertices
\State Remove all paths that are a single edge
\While{$P$ is not empty}
\While{$P$ is not empty}
\State Let $p$ be a path in $P$. Remove all edges in $p$ from $G$ and add them to $S$. 
\State Add the edge between the end points of $p$ to $S$ and remove it from $G$
\State Remove all paths from $P$ that intersect $p$. 
\EndWhile
\State $P$ = Shortest Paths between all adjacent vertices in new graph
\State Remove all paths that are a single edge
\EndWhile
\State \Return \Call{Verifier}{G,S}
\EndFunction
\end{algorithmic}
\end{algorithm}

Using the same proof Theorem \ref{thm:spc} we can show that {\sc General Short Path Cover} terminates after $O(k)$ iterations of the outer loop. Furthermore in this case do not need to show the distances are only increased. 

\subsection{5 Cycle Cover}
The algorithm presented in \cite{Raichel2018} has 3 major steps. The first two steps are used to estimate the support of the optimal solution and then the last step is actually used to find a solution given this support. We shall focus on the first 2 steps as these are where we are making some modifications. 

\begin{figure}[h]
\centering
\includegraphics[width = 0.3\textwidth, height = 0.4\textwidth]{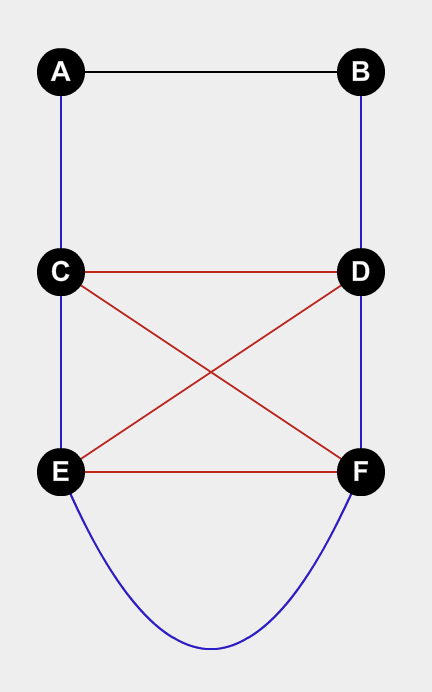} \hspace{3cm}
\includegraphics[width = 0.3\textwidth]{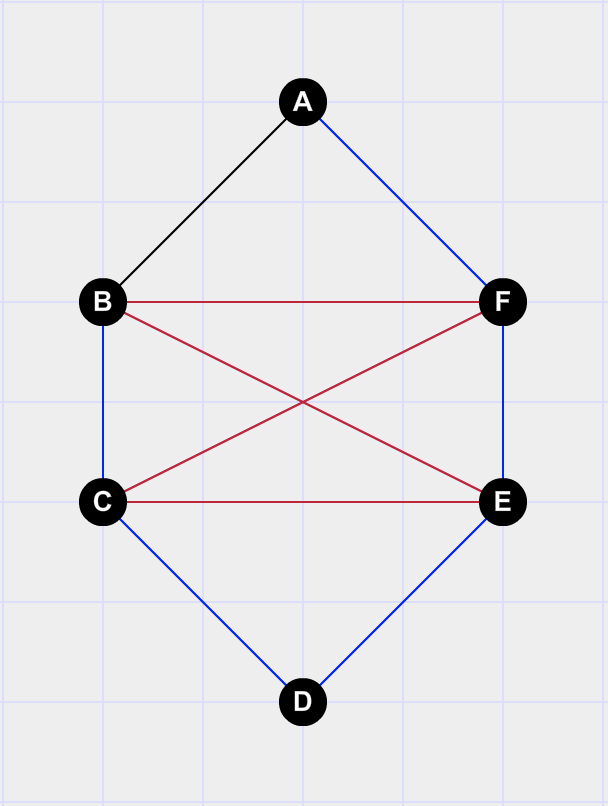}
\caption{Left: Embedding from \cite{Raichel2018}. Right: Our modified embedding for a smaller cycle. Here the black edge is the top edge. The blue edges are the bottom edges and the red edges are the embedded 4 cycle. The curved blue edge indicates that there are more vertices along that path}
\label{fig:4cycle}
\end{figure}

\textbf{First Step:} In the first step they estimate a support for all broken cycles of length $ \le m$. In particular they use the case when $m = 6$. As described in \cite{Raichel2018} we can $m-1$ approximation of the optimal cover for all broken cycles of length $ \le m$ in $O(n^m)$ time. Let us denote this cover by $S_{\le m}$. 

\textbf{Second Step:} For this step we need to first define unit cycles. Given a broken cycle $C$ with top edge $t$, let $e$ be a chord of $C$. Then $e$ divides $C$ into 2 cycles, one that contains $t$, denoted top($C,e$) and one that does not contain $t$ denoted bot($C,e$). Now we say this cycle is a unit cycle if for all chords $e$, $e$ is not the top edge of bot($C,e$). 

Then step 2 of the algorithm from \cite{Raichel2018} covered all unit cycles not covered by $S_{\le 6}$. Let $C$ be such a unit cycle. Now we know that $C$ has at least $7$ edges. Consider the red $C_4$ shown in Figure \ref{fig:4cycle}. Now we know that for each $e \in C_4$ we have that top$(C,e)$ is a broken cycle with less than or equal 6 edges. Hence, must have at least 1 edge in $S_{\le 6}$. But now since $C$ has no bottom edges in $S_{\le 6}$, we must have $e \in S_{\le 6}$. Thus, we know all edges in $C_4$ are edges in $S_{\le 6}$. They called this step $chord4(S_{\le 6})$.

For our modification we Figure  \ref{fig:4cycle} on the right shows how to embed the same 4 cycle in a 6 cycle instead of a 7 cycle. Thus, our modified algorithm is:

\begin{algorithm}[h]
\caption{5-Cycle Cover}
\begin{algorithmic}[1]
\Function{5 Cycle Cover}{G=(w,V)}
	\State Compute a regular cover of $S_{\le 5}$ of all broken cycles with $\le 5$ edges
	\State Compute a cover $S_c = chord4(S_{\le 5})$
	\State \Return \Call{Verifier}{$G,S_c \cup S_{\le 5}$}
	\EndFunction
\end{algorithmic}
\end{algorithm}

\subsection{IOMR-fixed}

The algorithm presented in \citet{Gilbert2017} is a s follows

\begin{algorithm}
\caption{IOMR Fixed}
\begin{algorithmic}[1]
\Require{$D \in \Sym$ }
\Function{IOMR-Fixed}{D}
\State $\hat{D} = D$
\For{$k \gets 1 \textrm{ to } n$}
\For{$i \gets 1 \textrm{ to } n$}
\State $\hat{D}_{ik} = \max(\hat{D}_{ik}, \max_{j < i}(\hat{D}_{ij}-\hat{D}_{jk}))$
\EndFor
\EndFor
\State \Return{$\hat{D}-D$}
\EndFunction
\end{algorithmic}
\end{algorithm}

We will now prove a few results about its worst case approximation ratio and its approximation ratio if {\rm OPT} = $\Theta(n^2)$.

\begin{lem} \label{lem:IOMRMax} {\sc IOMR-Fixed} can update or repair a maximum of $(n-1)(n-2)$ entries. \end{lem} 
\begin{proof} We will show that for every $s > 1$, there exists a $t < s$ such that $D_{st}$ is never updated which implies that the number of entries in the matrix not touched is at least 
\[ 
	n + 2(n-1).
\] 
This figure includes $n$ entries on the diagonal and 2$(n-1)$ from the entries seen on previous iterations (the factor of two comes from maintaining symmetry). Thus the maximum number of entries that can be updated is 
\[ 
	n^2 - n - 2(n-1) = n(n-1) - 2(n-1) = (n-1)(n-2).
\]

Let $s > 1$ be fixed and consider the entry $D_{s1}$. We see that {\sc IOMR-fixed} can only look at each entry $D_{ab}$ twice, once when $a=i, b=k$ and once when $a=k, b=i$. From now on we shall refer to these as the first and second time that we look at the entries. We can also see the first time occurs when $k = \min(a,b)$ and the second time when $k = \max(a,b)$. Thus, we see that entries of the form $D_{il}$ can be changed only the first time that they are seen. This is because the second time $D_{il}$ is seen, we have $i=1$. Hence, there are no indices $j < i$ at which we update. Thus, $D_{s1}$ can only be updated the first time it is seen (when $k=1$). Thus, there are two possibilities. 

\textbf{Case 1:} $D_{s1}$ is not updated the first time it is seen. Then, letting $t=1$ we get the result that we wanted. 

\textbf{Case 2:} $D_{st}$ is updated the first time it is looked at. Then we know that there must exist an $r < s$ such that \textbf{before} the update we have that \[ D_{sr} > D_{s1} + D_{1r}. \]

Now since $r < s$ we have already seen $D_{r1}$ for the first time. Thus at this point $D_{r1}$ is fixed and cannot be changed again. Now let us make the update after which we have \[ D_{rs} = \hat{D}_{s1} + D_{1r} \] where $\hat{D}_{r1} = D_{rs}-D_{ir}$. Now both terms on the right hand side have 1 as an index and have been seen once. Thus, they are now fixed and can never be updated again. 

Let us consider $D_{rs}$. Here we have that $r,s > 1$, this entry has not been seen yet, and it has not been updated. Now if it is updated in the future, because we are in the increase only case, updating this entry will break the above triangle. Since we know that {\sc IOMR-fixed} is a correct algorithm, $D_{rs}$ cannot be updated. Thus, letting $t=r$ we get the desired result. 
\end{proof}

\begin{lem} \label{lem:OmegaIOMR} For every $n$, there exists an input matrix $D$ such that {\sc IOMR-Fixed} repairs $(n-1)(n-2)$ entries while an optimal algorithm repairs only $(n-1)/2$. 
\end{lem}
\begin{proof} Consider a matrix $D$ where \[ D_{ij} = \begin{cases} 0 & \text{ if } i \neq 1, j \neq 1 \\ 2^i & \text{ if } j=1, i > 1 \\2^j & \text{ if } i=1, j > 1 \end{cases} \]

First, we claim that all entries of the form $D_{s1}$ will never be updated as entries will only be updated first time they are seen. Thus 
\begin{align*} 
D_{s1} &= \max(D_{s1}, \max_{t < s}(D_{s1} - D_{1t}))  \\
&= \max(2^s, \max_{t < s} (2^s - 2^t)) \\
& = 2^s 
\end{align*}

Now we just have to verify that the rest of the non-diagonal entries are updated. Let us look at the first time an entry $D_{rs}$ is updated. (Here $r < s$). Then we have that \begin{align*} \hat{D}_{rs} &= \max(D_{rs}, \max_{t < s}(D_{st} - D_{tr})) \\
&=\max_{t < s}(D_{st} - D_{tr}) &[D_{rs}=0] \\
&\ge D_{s1}-D_{1r} \\
&= 2^s - s^r \\
&> D_{rs} \end{align*} 

Thus all other non-diagonal entries will be updated the first time they are seen and we have the desired bound. \end{proof}

For dense optimal solutions let us consider the following gadget. For any given $k < n$, we can construct the following distance matrix $D$ as an input for the increase only sparse metric repair problem.  
\[ \left[ \begin{array}{c c c | c c c  } 0 & \hdots & 0 & * & \hdots & * \\ \vdots &  & \vdots & \vdots & & \vdots \\
0 & \hdots & 0 & * & \hdots & * \\ \hline
* & \hdots & * & & & \\
\vdots & & \vdots & & & \\
* & \hdots & * & &  \end{array} \right] \]

Were we set $D_{ij} = k+1-i$ for $i = 1, \hdots, k$ and $j > k$ (and make all other non diagonal entries equal to 1). This gadget has the desired properties we want as long as $k$ is smaller than $n/2$. It turns out that {\sc IOMR-Fixed} algorithm will always fix the entries in the two $*$ blocks and will never touch the entries in 0 block. 

\begin{thm} \label{thm:dense} For the above construction, if we let $\gamma = k/n$ and if $\gamma < 0.5$, then the optimal solution fixes $\gamma n  (\gamma n - 1)$ while {\sc IOMR-Fixed} repairs $2(n - \gamma n)(\gamma n - 1)$ entries. We can see that IOMR will correct all entries in the two $*$ blocks expect of entries of the form $1j$ or $j1$ and we obtain an approximation ratio of \[ \frac{2(n - \gamma n)(\gamma n - 1)}{\gamma n (\gamma n -1)} = 2\left(\frac{1}{\gamma} - 1 \right) \] 
\end{thm}
\begin{proof} Consider a broken triangle $D_{rs} > D_{rt}+D_{ts}$ (Note $r,s,t$ have to be distinct). Now if both $r,s \le k$ then we have that $D_{rs} = 0$ and the triangle could not be broken. Suppose then $r,s > k$. In this case $D_{rs} = 1$ and $D_{rt},d_{ts} \ge 1$. Thus, the triangle is not broken. 

Hence, it must be the case that exactly one of $r,s$ is bigger than $k$ and one is smaller. WLOG assume that $s > k$ and $r \le k$. Now if $t > k$. We have that $D_{rs} = D_{rt}$. Hence the triangle cannot be broken. Thus $t \le k$. Now if $t < r$. We have that \[ D_{rs} = k+1-r < k+1-t = D_{rt}.
\] 
Thus again the triangle could not have been broken.

Thus, looking through all the cases it is clear that the only broken triangles are of the form \[ D_{rs} > D_{rt} + D_{ts} \] Where $r < t \le k$ and $s > k$. In this case we have that $D_{rt} = 0$ and $D_{rs} = k+1-r$ and $D_{ts} = k+1-t$. Then since $0 < r < t$ we have that $k+1-r > k+1-t$. 

First let us see that changing all entries $D_{ij}$ where $i \le k$ and $j \le k$ to $\max(i,j)-\min(i,j)$ is a valid solution. This fixes all broken triangles of the above form as we now have \[ k+1-r = t-r +k +1 -t \Rightarrow D_{rs} = D_{rt}+D_{ts} \] Thus we fix all broken triangles. The only thing to make sure this is a valid solution is to make sure we don't break any new triangles.

Any new triangles would have to be of the form \[ D_{rs} > D_{rt}+D_{ts} \] where $r,s \le k$. Suppose if $t \le k$. We have that $D_{rs} = \max(r,s) - \min(r,s)$ and the right hand side is $\max(s,t)+\max(r,t) - \min(s,t) - \min(r,t)$. We can see for all six orderings that these triangles are not broken. Thus, let us now look at $t > k$. In this case we have that the right hand side is $2k+2-r-s$, while the left hand side is $|r-s|$. Then since $r,s \le k$ we see that these triangles cannot be broken either. Thus this is a valid solution. 

Finally we want to see that this solution is an optimal solution. We can see this using the previous Lemma. Suppose we have an optimal solution $P$ where we have zeros on $m$ different rows below the diagonal. In this case we see that at least all $m$ of these columns have to be equal. But this requires $2(m-1)(n-k)$ updates outside of the initial block of 0s. But since $k < n/2$, we see that it would be better to update the 0s instead. Thus, we see that the above solution is optimal. 

Finally to see the theorem, we just need to see what {\sc IOMR-fixed} does on the above input. We claim that {\sc IOMR-fixed} changes all entries in two $*$ blocks to $k+1$ and doesn't touch any other entry.  This can be seen via a simple induction argument on the outer loop of the {\sc IOMR-fixed} algorithm. \end{proof}

\end{document}